\documentclass[11pt,reqno]{amsart}

\usepackage{amssymb,amsmath,amsthm,amscd,latexsym,amsfonts}
\usepackage{mathtools}
\usepackage[T1]{fontenc}

\usepackage{graphicx}
\usepackage{color}
\usepackage[a4paper,top=3cm, bottom=3cm, left=3cm, right=2cm]{geometry}

\newtheorem{thm}{Theorem}
\newtheorem{defn}{Definition}
\newtheorem{lemma}{Lemma}

\newtheorem{rk}{Remark}
\newtheorem{cor}{Corollary}

\newtheorem{ex}{Example}

\numberwithin{equation}{section} 

\newcommand{\mypp}{\mbox{$\;\!$}}

\newcommand{\myn}{\mbox{$\;\!\!$}}

\newcommand{\bea}{\begin{eqnarray}}
\newcommand{\eea}{\end{eqnarray}}

\newcommand{\Z}{\mathbb{Z}}
\newcommand{\Q}{\mathbb{Q}}

\def\Z{\mathbb{Z}}

\begin{document}
\title[$p$-adic boundary laws and Markov chains]{$p$-adic boundary laws and Markov chains on trees}

\author{A. Le Ny, L. Liao,  U. A. Rozikov}

\address{A. \ Le Ny \\ Universit\'e Paris-Est, Laboratoire d'Analyse et de Math\'ematiques Appliqu\'ees, LAMA UMR CNRS 8050, UPEC, 91 Avenue du G\'en\'eral de Gaulle, 94010 Cr\'eteil cedex, France.}
\email {arnaud.le-ny@u-pec.fr}

\address{L. \ Liao\\ Universit\'e Paris-Est, Laboratoire d'Analyse et de Math\'ematiques Appliqu\'ees, LAMA UMR CNRS 8050, UPEC, 91 Avenue du G\'en\'eral de Gaulle, 94010 Cr\'eteil cedex, France.}
\email {lingmin.liao@u-pec.fr}

 \address{U.\ A.\ Rozikov\\ Institute of mathematics,
81, Mirzo Ulug'bek str., 100170, Tashkent, Uzbekistan.}
\email {rozikovu@yandex.ru}

\begin{abstract} In this paper we consider $q$-state potential on general infinite trees with a nearest-neighbor $p$-adic
interactions given by a stochastic matrix. {We show the uniqueness of the associated Markov chain ({\em splitting Gibbs measures}) under some sufficient conditions on the stochastic matrix.}
Moreover, we find a family of stochastic matrices for which there are at least two $p$-adic
Markov chains on an infinite tree (in particular, on a Cayley tree).
When the $p$-adic norm of $q$ is greater ({{\em resp.}} less) than the norm of any element of the stochastic matrix
then it is proved that the $p$-adic Markov chain is bounded ({{\em resp.}} is not bounded).
Our method {uses} a classical boundary law argument carefully adapted from
the real case to the $p$-adic case, by a systematic use of some
nice peculiarities of the ultrametric ($p$-adic) norms.

\end{abstract}
\maketitle



\section{Introduction} \label{sec:intro}

In this paper we develop a boundary law argument to study $p$-adic Markov chains on general trees.
In the real case Markov chains on trees are particular cases of Gibbs measures
corresponding to a Hamiltonian with nearest-neighbor interactions.
In the theory of Gibbs measures on trees (see \cite[Chapter 12]{Ge} and \cite{Ro})
the main problem is to describe the set of limiting Gibbs measures corresponding to a given Hamiltonian.
A complete analysis
of this set is often a difficult problem, this is even not completely described for
the Ising model (see \cite{GH, GMRS, GR, RR}  for some recent results).

 Parallel to the real valued Gibbs measures, the $p$-adic Gibbs measures are studied using
the $p$-adic mathematical physics
in \cite{7, 23, 24, RK, 48}. A $p$-adic distribution is an analogue of ordinary distributions
that takes values in a ring of  $p$-adic numbers \cite{22}, \cite{23}. 
 Analogically to a measure on a measurable space, a $p$-adic measure is a special case of a $p$-adic distribution. A $p$-adic distribution taking values in a normed space is called {\em{a $p$-adic measure}} if the values on compact open subsets are bounded.

It is known  that some $p$-adic models in physics cannot be described using ordinary Kolmogorov's probability
theory \cite{24, 29, 34, 48}. In \cite{28} the $p$-adic probability theory was developed using
the theory of non-Archimedean measures \cite{40}.
In \cite{GRR,19, 36, 37, 38, RT} various models of statistical physics
in the context of $p$-adic fields are studied.

In probability theory Kolmogorov's extension theorem (see, {\em e.g.}, \cite[Chapter~II,
\S\mypp3, Theorem~4, page~167]{44}), says that a compatibility
condition of a sequence of probability measures ensures that there exists a unique (limit) measure.
This theorem is used to introduce (real-valued) Markov chains on trees (see \cite[Chapter 12]{Ge}) by notion of a boundary law.
 A $p$-adic analogue of Kolmogorov's theorem was proved in
\cite{16}. Such a $p$-adic Kolmogorov theorem allows us to construct wide classes of
stochastic processes and to develop statistical
mechanics in the context of $p$-adic theory \cite{33}-\cite{38}.

In the present paper we introduce $p$-adic Markov chains on general infinite trees.
Such chains are constructed by $p$-adic boundary laws (for the real case see \cite[Chapter 12]{Ge}).  We also discuss the uniqueness  and boundedness of the $p$-adic Markov chain. The boundedness of the  $p$-adic measure is needed  to integrate $p$-adic valued functions  \cite{22, 23, 41}, and also to consider conditional expectations \cite{22,33}. Note that  $p$-adic measures are also useful in $p$-adic $L$-functions following
the works of B. Mazur (see  \cite{GG, 29}  for details).

 The paper is organized as follows. Section 2 presents definitions and known results.
 Section 3 is devoted to an introduction of $p$-adic Markov chains through boundary laws. Section 4 (resp. Section 5) is devoted
 to  finding a sufficient condition of  the uniqueness ({\em resp.} non-uniqueness) of $p$-adic Markov chain. In Section 6 we give some conditions ensuring
 that the $p$-adic Markov chain is ({\em resp.} not) bounded.

\section{Preliminaries}

\subsection{$p$-adic numbers and measures.} Let $\Q$ be the field of rational numbers. For a fixed prime number $p$, every rational number $x\ne 0$ can be represented
in the form $x = p^r{n\over m}$, where $r, n\in \Z$, $m$ is a positive integer, and $n$ and $m$ are relatively prime with $p$: $(p, n) = 1$, $(p, m) = 1$. The $p$-adic norm of $x$ is given by
$$|x|_p=\left\{\begin{array}{ll}
p^{-r}\ \ \mbox{for} \ \ x\ne 0\\
0\ \ \mbox{for} \ \ x = 0.
\end{array}\right.
$$
This norm is non-Archimedean  and satisfies the so-called {\em strong triangle inequality}
$$|x+y|_p\leq \max\{|x|_p,|y|_p\}.$$

We will often use the following fact:
\begin{equation}\label{dg}
{\rm If}\ \ |x|_p\ne |y|_p \ \ {\rm then} \ \ |x+y|_p=\max\{|x|_p,|y|_p\}.
\end{equation}

The completion of $\Q$ with respect to the $p$-adic norm defines the $p$-adic field
 $\Q_p$. Any $p$-adic number $x\ne 0$ can be uniquely represented
in the canonical form
\begin{equation}\label{ek}
x = p^{\gamma(x)}(x_0+x_1p+x_2p^2+\dots),
\end{equation}
where $\gamma(x)\in \Z$ and the integers $x_j$ satisfy: $x_0 > 0$, $0\leq x_j \leq p - 1$ (see
\cite{29,41,48}). In this case $|x|_p = p^{-\gamma(x)}$.\\

Our analysis will strongly relies on nice properties of the $p$-adic norm, and on the two following classical results in $p$-adic algebra.

\begin{thm}[\cite{29, 48}]\label{tx2}
The equation
$x^2 = a$, $0\ne a =p^{\gamma(a)}(a_0 + a_1p + ...), 0\leq a_j \leq p - 1$, $a_0 > 0$
has a solution $x\in \Q_p$ if and only if the following conditions
are fulfilled:

i) $\gamma(a)$ is even;

ii) $a_0$ is a quadratic residue modulo $p$ if $p\ne 2$; $a_1 = a_2 = 0$ if $p = 2$.
\end{thm}

The elements of the set $\mathbb{Z}_p=\{x\in \Q_p: |x|_p\leq 1\}$ are called $p$-adic integers.\\

The following statement is known as {\em Hensel's lemma} \cite[Theorem 3.15]{AK}.

\begin{thm}\label{hl} Let $F(x)=\sum_{i=0}^nc_ix^i$ be a polynomial whose coefficients are $p$-adic integers. Let
 $F'(x)=\sum_{i=0}^nic_ix^{i-1}$ be the derivative of $F(x)$. Assume there exist $a_0\in \mathbb{Z}_p$
 and $\gamma\in \{0,1,2,\dots\}$ such that
 $$\begin{array}{lll}
 F(a_0)\equiv 0 \,(\operatorname{mod} p^{2\gamma+1}),\\[2mm]
 F'(a_0)\equiv 0 \,(\operatorname{mod} p^{\gamma}),\\[2mm]
 F'(a_0)\neq 0 \,(\operatorname{mod} p^{\gamma+1}).
 \end{array}$$
  Then there exists $a\in \mathbb{Z}_p$ such that
 $F(a)=0$ and $a\equiv a_0\,(\operatorname{mod} p^{\gamma+1})$.
 \end{thm}

Given $a\in \Q_p$ and $r > 0$ put
$$B(a, r) = \{x\in \Q_p : |x-a|_p < r\}.$$

The $p$-adic {\it logarithm} is defined by the series
$$\log_p(x) =\log_p(1 + (x-1)) =
\sum_{n=1}^\infty (-1)^{n+1}{(x-1)^n\over n},$$
which converges for $x\in B(1, 1)$; the $p$-adic exponential is defined by
$$\exp_p(x) =\sum^\infty_{n=0}{x^n\over n!},$$
which converges for $x \in B(0, p^{-1/(p-1)})$.


\begin{lemma}[\cite{29}]\label{l1} Let $x\in B(0, p^{-1/(p-1)})$, then
$$|\exp_p(x)|_p = 1,\ \ |\exp_p(x)-1|_p = |x|_p, \ \ |\log_p(1 + x)|_p = |x|_p,$$
$$\log_p(\exp_p(x)) = x,\ \ \exp_p(\log_p(1 + x)) = 1 + x.$$
\end{lemma}

Let $(X,{\mathcal B})$ be a measurable space, where ${\mathcal B}$ is an algebra of subsets of $X$. A function $\mu: {\mathcal B}\to \Q_p$
is said to be a $p$-adic measure if for any $A_1, . . . ,A_n\in {\mathcal B}$ such that $A_i\cap A_j = \emptyset$, $i\ne j$, the following holds:
$$\mu(\bigcup^n_{j=1}A_j)=\sum^n_{j=1}\mu(A_j).$$
A $p$-adic measure is called a {\em $p$-adic probability measure} if $\mu(X) = 1$, see, {\em e.g.} \cite{22, 40}. Let us warn that due to the different axiomatic and ring of values, some intuitive properties of sets of probability measures (like {\em e.g.} some convex properties) are not valid anymore \cite{RK}.

\subsection{Tree.} A tree is a connected graph without cycles (see \cite{Rtr} for more details).
Let $\mathcal T=(V, L)$ be a tree, where $V$ is the set of vertices and  $L$ is the set of edges.
Two vertices $x$ and $y$ are called {\it nearest neighbors} if there exists an
edge $b \in L$ connecting them.
We will use the notation $b=\langle x,y\rangle$ for the edge connecting the vertices $x$ and $y$.
A collection of nearest neighbor pairs $\langle x,x_1\rangle, \langle x_1,x_2\rangle,...,\langle x_{d-1},y\rangle$ is called a {\it
path} from $x$ to $y$. The distance $d(x,y)$ on the tree is the number of edges of the shortest path from $x$ to $y$.

For $z\in V$, we denote
$$L^z=\{\langle x,y\rangle\in L: d(z,x)=d(z,y)+1\},$$
$$\prescript{z}{}{L}=\{\langle x,y\rangle\in L: d(z,y)=d(z,x)+1\}.$$
Let $A\subset V$. Denote
$$\partial A=\{x\in V\setminus A: \ \exists y\in A, \ \ \langle x, y \rangle\}.$$

\section{$p$-adic Markov chain and boundary laws} We consider a system with nearest
 neighbor interactions on a tree
where the spins assigned to the vertices of the tree take values in the set $\Phi:=\{1,2,\dots, q\}$.

A configuration $\sigma_A$ on $A\subset V$ is then defined as a function $x\in A\mapsto\sigma_A(x)\in\Phi$.
The set of all configurations is $\Phi^A$.

By $p$-adic probability vector we mean a vector
with $p$-adic valued coordinates summing to 1.
A $p$-adic stochastic matrix is a matrix with each row being a $p$-adic probability vector.

 For each edge $b=\langle x, y\rangle\in L$ we consider a stochastic
matrix $\mathbb P_b=\left(P_b(i,j) \right)_{i,j=1}^q$. For each $x\in V$
consider a probability vector $\alpha_x=(\alpha_{1,x}, \dots, \alpha_{q,x})$.

For any edge $b=\langle x, y\rangle\in L$  we assume that
\begin{equation}\label{cs}
\alpha_{i,x}P_b(i,j)=\alpha_{j,y}P_b(j,i), \ \ \forall i,j\in \Phi.
\end{equation}

\begin{defn} A $p$-adic probability distribution (measure) $\mu$ is
called a $p$-adic Markov chain with transition matrices $(P_b)_{b\in L}$ and marginal distribution
$\alpha_x$ at $x\in V$ if for all finite, connected set $\Lambda \subset V$, and all $\zeta\in \Phi^\Lambda$
and $z\in \Lambda$ the following holds
\begin{equation}\label{pm}
\mu(\sigma_\Lambda=\zeta)=\alpha_z(\zeta_z)\prod_{{\langle x, y\rangle\in ^zL:\atop x,y\in \Lambda}}P_{\langle x, y\rangle}(\zeta_x,\zeta_y).
\end{equation}
\end{defn}

Note that the reversibility condition (\ref{cs}) is equivalent to the statement that
the expression on the right of (\ref{pm}) is independent of the choice of $z\in \Lambda$.

Consider for each edge $b=\langle x, y \rangle$ a
matrix  $\mathcal Q_b=(Q_b(i,j))_{i,j=1}^q$.  We always assume
\begin{align}\label{condition-matrix}
\begin{array}{lll}
Q_{\langle x, y \rangle}(i,j)=Q_{\langle y, x \rangle}(j,i),\\[3mm]
 \sum_{j=1}^{q}Q_{\langle x, y \rangle}(j,i)=1.
 \end{array}
\end{align}

Let $\boldsymbol{z}(x,y)=(z_1(x,y), \dots, z_q(x,y))$ be a vector in $\mathbb{Q}_p$.

\begin{defn} For $(\mathcal Q_b)_{b\in L}$ satisfying (\ref{condition-matrix}), a \emph{$p$-adic boundary law}\footnote{Compare with real boundary law
of \cite[Definition~(12.10)]{Ge}.} $\{\boldsymbol{z}(x,y)\}_{\langle x,y\rangle \in
L}$ is such that for any $\langle
x,y\rangle\in L$, and for all $i\in\Phi$, it holds
\begin{equation}\label{eq0}
z_i(x,y)=c(x,y)\prod_{v\in\partial\{x\}\setminus\{y\}}\sum_{j\in\Phi}
z_j(v,x)Q_{\langle v, x \rangle}(j,i),
\end{equation}
where $c(x,y)$ is an arbitrary constant (not depending on
$i\in\Phi$).
\end{defn}

Using (\ref{condition-matrix}) and proceeding as in the classical case of \strut{}\cite[Formula~(12.13), page~243]{Ge}, one directly gets that each boundary law
$$\boldsymbol{z}=\{\boldsymbol{z}(x,y)=(z_1(x,y), \dots, z_q(x,y))\}_{\langle x, y \rangle\in L}$$
defines a $p$-adic Markov chain $\mu^{\boldsymbol{z}}$:  for any finite \emph{connected} set
$\varnothing\ne\varLambda\subset V$ (and
$\bar{\varLambda}=\varLambda\cup\partial\varLambda$), one has
\begin{equation}\label{eq:mu-h-ex-Lambda}
\mu^{\boldsymbol{z}}(\sigma_{\myn\bar{\varLambda}}\myn=\varsigma)=\frac{1}{Z_{\myn\bar{\varLambda}}}\prod_{x\in \partial \Lambda}
{z}_{\varsigma(x)}(x,x_\Lambda)\prod_{{b\in L:\atop b\cap \Lambda\ne \emptyset}}Q_b(\varsigma_b),
\end{equation}
where $Z_{\bar{\varLambda}}=Z_{\bar{\varLambda}}(\boldsymbol{z})$ is the normalizing factor, $x_{\myn\varLambda}$
denotes the unique neighbor of $x\in\partial\varLambda$ belonging
to $\varLambda$, and $\varsigma_b=(\varsigma(u), \varsigma(v))$, for $b=\langle u, v\rangle$.
We stress that the first condition in (\ref{condition-matrix}), which is \cite[Formula~(12.9)]{Ge}, 
is needed to check that $\mu^{\boldsymbol{z}}$ is a well defined $p$-adic Markov chain.

\section{{Criterion for uniqueness of the $p$-adic Markov chain}}

A $p$-adic Markov chain can be considered as a particular case of $p$-adic Gibbs measure defined through the $p$-adic
exponential $\exp_p(x)$, with $|x|_p<  p^{-1/(p-1)}$ \cite{36}.
As it was mentioned above, the set of values of a $p$-adic norm $|\cdot|_p$ is $\{p^m: m\in \Z\}$, {so}
the condition $|x|_p<  p^{-1/(p-1)}$ is equivalent to the condition $|x|_p\leq {1\over p}$.
Consequently, {we shall restrict part of the analysis to quantities belonging} to the set:
$$\mathcal E_p=\left\{x\in \Q_p: |x-1|_p\leq {1\over p}\right\}.$$

The following lemma {will also be} useful (see \cite[Lemma 4.6]{36}).

\begin{lemma}\label{l2} If $a_i\in  \Q_p$  for all $i=1,\dots,m$ are such that
$$|a_i|_p=1, \ \ \ |a_i-1|_p\leq M,$$
then
$$\left|\prod_{i=1}^ma_i\right|_p=1, \ \ \left|\prod_{i=1}^ma_i-1\right|_p\leq M.$$
\end{lemma}

Without loss of generality, we set hereafter $z_{q}(v,x)\equiv 1$ (a normalization at $q$).
Then the condition (\ref{eq0}) for the stochastic matrix $\mathcal Q_b=(Q_b(i,j))_{i,j=1}^q$
reads
\begin{equation}\label{e1}
z_i(x,y)=\prod_{v\in\partial\{x\}\setminus\{y\}}{1+\sum_{j=1}^{q-1}(z_j(v,x)-1)Q_{\langle v, x \rangle}(j,i)\over
1+\sum_{j=1}^{q-1}(z_j(v,x)-1)Q_{\langle v, x \rangle}(j,q)}, \ \ i=1,2,\dots, q-1.
\end{equation}
Here we have used
$$Q_{\langle v, x \rangle}(q,i)=1-\sum_{j=1}^{q-1}Q_{\langle v, x \rangle}(j,i), \ \ i=1,2,\dots ,q.$$

 In this section we examine the conditions on the parameters $k\geq 1$ and on $\mathcal Q_b$ for the existence and the uniqueness of the solutions of the equation (\ref{e1}).

For the uniqueness, we ssume that the matrix $\mathcal Q_b=(Q_b(i,j))_{i,j=1}^q$ satisfies the following conditions
\begin{equation}\label{cq}
\begin{array}{lll}
|Q_{\langle x, y \rangle}(j,i)|_p\leq 1, \\[3mm]
\left|Q_{\langle x, y \rangle}(j,i)-Q_{\langle x, y \rangle}(j,q)\right|_p\leq {1\over p},  \ \ \forall \langle x, y \rangle, \ \ \forall i,j.
\end{array}
\end{equation}

\begin{thm}\label{t2} Assume each vertex of the tree has degree {at least $2$ and that the matrix} $\mathcal Q_b=(Q_b(i,j))_{i,j=1}^q$ satisfies (\ref{condition-matrix}) and (\ref{cq}).
Then the equation (\ref{e1}) has a unique solution $\boldsymbol{z}(x,y)\equiv (1,1,\dots,1)\in {\mathcal E}^{q-1}_p$, $\langle x, y \rangle\in L$.

\end{thm}
\begin{proof} Since
$$\sum_{j=1}^{q}Q_{\langle v, x \rangle}(j,i)=1, \ \ \forall \langle v, x \rangle, \ \ \forall i,$$
it follows that  $\boldsymbol{z}(x,y)\equiv (1,1,\dots,1)$ is a solution to (\ref{e1}).

We show its uniqueness. For  $z=(z_1,\dots, z_{q-1})\in \Q_p^{q-1}$, we introduce the norm
$$\|z\|=\max_{i}|z_i|_p.$$

Let $\boldsymbol{z}(x,y)\in {\mathcal E}^{q-1}_p$, $\langle x, y \rangle\in L$
be a solution.
Denote
\begin{equation}\label{kk}
K_i\equiv K_i(v,x,q)={1+\sum_{j=1}^{q-1}(z_j(v,x)-1)Q_{\langle v, x \rangle}(j,i)\over
1+\sum_{j=1}^{q-1}(z_j(v,x)-1)Q_{\langle v, x \rangle}(j,q)}.
\end{equation}
Using (\ref{dg}), (\ref{cq}) and Lemma \ref{l2}, we calculate $|K_i|_p$:
$$|K_i|_p=\left|{1+\sum_{j=1}^{q-1}(z_j(v,x)-1)Q_{\langle v, x \rangle}(j,i)\over
1+\sum_{j=1}^{q-1}(z_j(v,x)-1)Q_{\langle v, x \rangle}(j,q)}\right|_p=1.$$
Let us now estimate $|K_i-1|_p$  using (\ref{cq}):
\begin{align*}
|K_i-1|_p&=\left|{\sum_{j=1}^{q-1}[z_j(v,x)-1]\{Q_{\langle v, x \rangle}(j,i)-Q_{\langle v, x \rangle}(j,q)\}\over
1+\sum_{j=1}^{q-1}(z_j(v,x)-1)Q_{\langle v, x \rangle}(j,q)}\right|_p\\
&=\left|\sum_{j=1}^{q-1}[z_j(v,x)-1]\{Q_{\langle v, x \rangle}(j,i)-Q_{\langle v, x \rangle}(j,q)\}\right|_p\\
&\leq \max_{j}\left|z_j(v,x)-1\right|_p\left|Q_{\langle v, x \rangle}(j,i)-Q_{\langle v, x \rangle}(j,q)\right|_p\\
&\leq {1\over p}\|\boldsymbol{z}(v,x)-1\|\leq {1\over p} \|\boldsymbol{z}(\hat v,x)-1\|,
\end{align*}
 where we have used the  hypothesis
 $$\left|Q_{\langle v, x \rangle}(j,i)-Q_{\langle v, x \rangle}(j,q)\right|_p\leq {1\over p},$$
 and $\hat v\equiv \hat v(x,y)$ is defined by
 $$\|\boldsymbol{z}(\hat v,x)-1\|= \max_{v\in\partial\{x\}\setminus\{y\}} \|\boldsymbol{z}(v,x)-1\|.$$
 Thus $K_i$ satisfies the conditions of Lemma \ref{l2}, and we have
\begin{align}\label{e2}
\begin{split}
|z_i(x,y)-1|_p=&\left|\prod_{v\in\partial\{x\}\setminus\{y\}}{1+\sum_{j=1}^{q-1}(z_j(v,x)-1)Q_{\langle v, x \rangle}(j,i)\over
1+\sum_{j=1}^{q-1}(z_j(v,x)-1)Q_{\langle v, x \rangle}(j,q)}-1\right|_p\\
=&\left|\prod_{v\in\partial\{x\}\setminus\{y\}}K_i-1\right|_p\leq {1\over p} \|\boldsymbol{z}(\hat v,x)-1\|.
\end{split}
\end{align}
Consequently,
\begin{equation}\label{zb}
\|\boldsymbol{z}(x,y)-1\|\leq {1\over p}\|\boldsymbol{z}(\hat v,x)-1\|.
\end{equation}
Since this estimation is true for arbitrary edge $\langle x, y \rangle\in L$,
one can start from any edge and then iterate the estimation (\ref{zb}), to obtain the following
\begin{equation}\label{zb1}
\|\boldsymbol{z}(x,y)-1\|\leq {1\over p^n}\|\boldsymbol{z}(\hat v^{(n)},\hat v^{(n-1)})-1\|\leq {1\over p^{n+1}}.
\end{equation}
which as $n\to\infty$ gives $\boldsymbol{z}(x,y)\equiv 1.$
\end{proof}
Denote by $\mu^{\boldsymbol{1}}$ the $p$-adic Markov chain which corresponds to $\boldsymbol{z}(x,y)\equiv (1,\dots,1).$
\begin{cor} Under the conditions of Theorem \ref{t2}, there exists a unique $p$-adic Markov chain, which  satisfies that
for any finite \emph{connected} set
$\varnothing\ne\varLambda\subset V$ (and
$\bar{\varLambda}=\varLambda\cup\partial\varLambda$),
\begin{equation}\label{bda}
\mu^{\boldsymbol{1}}(\sigma_{\myn\bar{\varLambda}}\myn=\varsigma)=\frac{1}{Z_{\myn\bar{\varLambda}}}\prod_{{b\in L:\atop b\cap \Lambda\ne \emptyset}}Q_b(\varsigma_b),
\end{equation}
where
\begin{equation}\label{dz}
Z_{\myn\bar{\varLambda}}=\sum_{\sigma\in \Omega_\Lambda}\prod_{{b\in L:\atop b\cap \Lambda\ne \emptyset}}Q_b(\sigma_b).
\end{equation}
\end{cor}

\section{{Criterion for non-uniqueness of the $p$-adic Markov chains}}
\subsection{On a regular tree.}
Consider the Cayley tree of order $k\geq 1$. Suppose the matrix $\mathcal Q_b$ in the system of equations (\ref{e1}) satisfies the condition
\begin{equation}\label{qc}
Q_{\langle v, x \rangle}(1,i)=Q_{\langle v, x \rangle}(1,q), \ \ \forall \langle v, x \rangle, \ \ i=2,\dots,q-1.
\end{equation}
We assume further that $Q_{\langle v, x \rangle}(1,1)$ and $Q_{\langle v, x \rangle}(1,q)$ are independent on $\langle v, x \rangle$, that is
\begin{equation}\label{qq1}
\alpha\equiv Q_{\langle v, x \rangle}(1,1), \ \ \beta\equiv Q_{\langle v, x \rangle}(1,q), \forall \langle v, x \rangle\in L.
\end{equation}
\begin{thm}\label{t3} If (\ref{qc}), (\ref{qq1}) are satisfied, $\alpha, \beta$ are $p$-adic integers,  and there exists $\gamma\in \{0,1,2,\dots\}$
such that
\begin{align}\label{condition-Hensel}
 \begin{array}{lll}
 k(\beta-\alpha)+1 \equiv 0 \,(\operatorname{mod} p^{2\gamma+1}),\\[2mm]
 k\beta+{k(k-1)\over 2}(\beta^2-\alpha^2)\equiv 0 \,(\operatorname{mod} p^{\gamma}),\\[2mm]
k\beta+{k(k-1)\over 2}(\beta^2-\alpha^2)\neq 0 \,(\operatorname{mod} p^{\gamma+1}),
 \end{array}
 \end{align}
then the equation (\ref{e1}) has at least two solutions.
\end{thm}
\begin{proof} We shall prove that  the equation (\ref{e1}) has two constant (translational-invariant) solution $\boldsymbol{z}(x,y)\equiv \boldsymbol{z}$, $\forall \langle x, y\rangle \in L$. The first solution is already known: $\boldsymbol{z}(x,y)\equiv (1,\dots,1)$.
We shall show that the system (\ref{e1}) has a solution of the following form
 $$\boldsymbol{z}=\{\boldsymbol{z}(x,y)=(z, 1,1, \dots, 1)\}_{\langle x, y \rangle\in L}, \ \ z\ne 1.$$
 Then from (\ref{e1}), for the Cayley tree of order $k\geq 2$, we get
 \begin{equation}\label{ef}
z=\left({1-\alpha+\alpha z\over
1-\beta+\beta z}\right)^k.
\end{equation}
 Independently on parameters, this equation has solution $z=1$. We are going
 to find conditions on $\alpha\ne\beta$ and on $k$ to have at least one solution $z\ne 1$.

The equation  (\ref{ef}) can be written as $F(z)=0$ with
$$F(z)=z(1-\beta+\beta z)^k-(1-\alpha+\alpha z)^k.$$
We are interested in the solution of $G(z)={F(z)\over z-1}=0$, where
$$G(z)=1+\sum_{j=1}^k{k\choose j}(z\beta^j-\alpha^j)(z-1)^{j-1}. $$
Since $\alpha, \beta$ are $p$-adic integers, $G(z)$ has only $p$-adic integer coefficients.
Now we shall check the other conditions of Hensel's lemma (see Theorem \ref{hl}).
Take $a_0=1$. Then we have $G(1)=1+k(\beta-\alpha)$ and
$$
G'(1)=k\beta+{k(k-1)\over 2}(\beta^2-\alpha^2).$$
Therefore  by (\ref{condition-Hensel}), the conditions of Hensel's lemma are satisfied for
 $G(z)$. Hence there exists a $p$-adic integer $a$ such that
 $G(a)=0$ and $a\equiv a_0\,(\operatorname{mod} p^{\gamma+1})$, i.e. $G(z)=0$ has a solution $z=a$.
Since $a_0=1$, we have $a\equiv 1\,(\operatorname{mod} p^{\gamma+1})$. Thus $a\in {\mathcal E}_p$.
This proves the theorem.
\end{proof}

\begin{rk} Note that if $p$ divides $k(\beta-\alpha)+1$ then $p$ does not divide $\beta-\alpha$,
therefore $|\beta-\alpha|_p=1>{1\over p}$, i.e. the condition (\ref{cq}) is not satisfied.

\end{rk}

Let us give some examples of parameters satisfying the conditions of Theorem \ref{t3}:

\begin{ex} The case $\gamma=0$:
\begin{itemize}
\item[a)] Let $k=1$. Then the equation $G(z)=0$ has a unique solution $z=a={\alpha-1\over \beta}$.  The condition (\ref{condition-Hensel}) of Theorem \ref{t3} is equivalent to
$$|\beta-\alpha+1|_p\leq {1\over p}, \ \ |\beta|_p=1.$$
This implies $|\alpha-1|_p=1$, and $|a|=1$, $|a-1|\leq {1\over p}$. Thus, the solution $z=a$, other than the solution $z=1$, is also in $\mathcal E_p$.

\item[b)] Take $k=2$, $p=3$, $\alpha=2$, $\beta=3$.
 Then $k(\beta-\alpha)+1=3\equiv 0\ (\rm{mod} \ 3)$ and
 $k\beta+{k(k-1)\over 2}(\beta^2-\alpha^2)=11\not\equiv 0\ (\rm{mod} \ 3)$.
 For these parameters the equation (\ref{ef}) has three solutions:
 $$z_0=1, \ \ z_1={7-\sqrt{13}\over 18},  \ \ z_2={7+\sqrt{13}\over 18}.$$
 Note (see Theorem \ref{tx2}) that $\sqrt{13}$ exists in $\Q_3$.  Moreover, it can be calculated\footnote{http://www.numbertheory.org/php/p-adic.html}:
 $$\sqrt{13}=1+2\cdot 3+3^2+3^5+2\cdot 3^6+\dots.$$
Then we get
$$|z_1|_3=\left|{7-\sqrt{13}\over 18}\right|_3=\left|{3^2+3^5+2\cdot 3^6+\dots\over 2\cdot 3^2}\right|_3=1,$$
$$|z_1-1|_3=\left|{-11-\sqrt{13}\over 18}\right|_3=\left|{3^3+3^5+2\cdot 3^6+\dots\over 2\cdot 3^2}\right|_3={1\over 3}.$$
Hence $z_1\in \mathcal E_3$, and $z_1$ plays the role of $a\in \mathcal E_3$ mentioned in the proof of Theorem \ref{t3}.
 On the other hand, we have $z_1z_2={1\over 9}$. Consequently $|z_1z_2|_3=9$. Since $|z_1|_3=1$, we obtain $|z_2|_3=9$.
Thus $z_2\notin \mathcal E_3$.
\end{itemize}
\end{ex}
\begin{ex}
The case $\gamma=1$: Take $k=2$, $p=3$, $\alpha=6$, $\beta=19$. Then
$$\begin{array}{lll}
 k(\beta-\alpha)+1=27 \equiv 0 \,(\operatorname{mod} 3^{3}),\\[2mm]
 k\beta+{k(k-1)\over 2}(\beta^2-\alpha^2)=363\equiv 0 \,(\operatorname{mod} 3),\\[2mm]
k\beta+{k(k-1)\over 2}(\beta^2-\alpha^2)=363\not\equiv 0 \,(\operatorname{mod} 3^2).
 \end{array}$$
 In this case the equation (\ref{ef}) has three solutions:
 $$z_0=1, \ \ z_1={359-39\sqrt{61}\over 722},  \ \ z_2={359+39\sqrt{61}\over 722}.$$
 We have $|z_1-1|_3=|{363-39\sqrt{61}\over 722}|_3\leq {1\over 3}.$ Thus $z_1\in \mathcal E_3$.
 Similarly, one can see that $z_2\in \mathcal E_3$.
\end{ex}

As a corollary of Theorem \ref{t3}, we have the following.

\begin{thm}\label{t3a} If  the conditions of Theorem \ref{t3} are satisfied then for the matrix $\mathcal Q_b$
 on the Cayley tree of order $k\geq 1$, there are at least two $p$-adic Markov chains.
\end{thm}
\begin{rk} Theorem \ref{t3} can be generalized as follows: fix $m\in \{1,2,\dots,q-1\}$ and assume
 \begin{equation}\label{qcm}
Q_{\langle v, x \rangle}(j,i)=Q_{\langle v, x \rangle}(j,q), \ \ \forall \langle v, x \rangle, \ \ j=1,\dots,m; \ \ i=m+1,\dots,q-1.
\end{equation}
Suppose  $\sum_{j=1}^mQ_{\langle v, x \rangle}(j,i)$ and $\sum_{j=1}^mQ_{\langle v, x \rangle}(j,q)$ are independent on $\langle v, x \rangle$, i.e.,
\begin{equation}\label{qq}
A\equiv \sum_{j=1}^mQ_{\langle v, x \rangle}(j,i), \ \ B\equiv \sum_{j=1}^mQ_{\langle v, x \rangle}(j,q), \ \ \forall \langle v, x \rangle\in L, \ \ i=1,\dots, m.
\end{equation}
Under the above mentioned conditions one can show that the system (\ref{e1}) has a solution of the following form
 $$\boldsymbol{z}=\{\boldsymbol{z}(x,y)=(\underbrace{z,z,\dots,z}_m, 1,1, \dots, 1)\}_{\langle x, y \rangle\in L}, \ \ z\ne 1.$$
 Then from (\ref{e1}), for the Cayley tree of order $k\geq 2$, we get
 \begin{equation}\label{efa}
z=\left({1-A+A z\over
1-B+B z}\right)^k.
\end{equation}
This equation is identical with (\ref{ef}) and it has non-unique solutions when $A$ and $B$ (replacing $\alpha$ and $\beta$) satisy the conditions mentioned in Theorem \ref{t3}.
\end{rk}

\subsection{Extension on a non-regular tree.}
Consider now a general tree $\mathcal T$, with each vertex having at least \emph{two} nearest neighbors.
Recall that $L$ is the set of all edges of $\mathcal T$.
Such a tree contains a Cayley tree (of some order $k\geq 1$) as a subtree, which we denote by $\Gamma^k$.
Let $L_k$ be the set of all edges of $\Gamma^k$, i.e., $L_k\subset L$.

Assume on $\Gamma^k$, the conditions of Theorem \ref{t3} are satisfied. Then we have a boundary law of the form
\begin{equation}\label{ct}
\boldsymbol{z}=\{\boldsymbol{z}(x,y)=(z, 1,1, \dots, 1)\}_{\langle x, y \rangle\in L_k}, \ \ z\ne 1.
\end{equation}
Let $g(z)={1-\alpha+\alpha z\over 1-\beta+\beta z}.$ Define on the edges $\langle x, y\rangle$ of the general tree $\mathcal T$
the following vector-valued function
\begin{equation}\label{gt}
\boldsymbol{l}=\{\boldsymbol{l}(x,y)=(l_1(x,y), 1, \dots, 1)\},
\end{equation}
where
\begin{equation}\label{gtt}
l_1(x,y)=\left\{\begin{array}{lll}
z, \ \ \ \ \ \ \ \ \mbox{if} \ \  \langle x, y\rangle\in L_k,\\[3mm]
1, \ \ \ \ \ \ \ \ \mbox{if} \ \  \langle x, y\rangle\in L, \ \ x\in L\setminus L_k,\\[3mm]
zg(z), \ \ \mbox{if} \ \  \langle x, y\rangle\in L, \ \ x\in\Gamma^k, \ \ y\in L\setminus L_k,
\end{array}\right.
\end{equation}
and $z$ is defined in (\ref{ct}).

For $i=2,\dots,q-1$, we assume
\begin{equation}\label{qcc}
Q_{\langle v, x \rangle}(1,i)=Q_{\langle v, x \rangle}(1,q), \ \ \mbox{for} \ \ \langle v, x \rangle \ \ \mbox{with} \ \
 v\in\Gamma^k, \ \ x\in L\setminus L_k.
\end{equation}
and show that $\boldsymbol{l}$ defined by (\ref{gt}) satisfies the equation (\ref{e1}).

For coordinates $l_i(x,y)=1$, $i=2,3,\dots,q-1$, from (\ref{e1}) we have
\begin{equation}\label{ee1}
1=l_i(x,y)=\prod_{v\in\partial\{x\}\setminus\{y\}}{1+(l_1(v,x)-1)Q_{\langle v, x \rangle}(1,i)\over
1+(l_1(v,x)-1)Q_{\langle v, x \rangle}(1,q)}, \ \ i=2,\dots, q-1.
\end{equation}
Therefore, by (\ref{qc}), (\ref{gtt}) and (\ref{qcc}), one can see that the right-hand side of (\ref{ee1}) is always 1.

Now we show that $l_1(x,y)$ also satisfies (\ref{e1}).
Indeed, we note that
$\partial\{x\}\setminus\{y\}=A_k(x,y)\cup B_k(x,y)$, where $A_k(x,y)=(\partial\{x\}\setminus\{y\})
\cap L_k$ and $B_k(x,y)=(\partial\{x\}\setminus\{y\})\cap (L\setminus L_k)$.

 We thus have the following three possible cases:

{\it Case: $x,y\in \Gamma^k$}. In this case $l_1(x,y)=z$ and $A_k(x,y)$ has $k$ elements. Therefore,
the equation (\ref{e1}) for $l_1(x,y)$ is reduced to $z=(g(z))^k$, which is satisfied by the conditions of Theorem \ref{t3}.

{\it Case: $\langle x, y\rangle\in L, \ \ x\in L\setminus L_k$}.
Then $A_k(x,y)=\emptyset$ and hence the equation (\ref{e1}) for $l_1(x,y)$ is reduced to the identity $1=1$.

{\it Case: $\langle x, y\rangle\in L, \ \ x\in\Gamma^k, \ \ y\in L\setminus L_k$.} In this case $A_k(x,y)$ contains
$k+1$ elements, and we have $l_1(v,x)=z$ for all $v\in A_k(x,y)$. Thus the equation (\ref{e1}) has the form
$l_1(x,y)=(g(z))^{k+1}$. Using $z=(g(z))^k$, we get $l_1(x,y)=zg(z)$ as in the definition (\ref{gt}).
Thus $\boldsymbol{l}(x,y)$ satisfies the equation (\ref{e1}).

Denote by $\mu^{\boldsymbol{l}}$ the  $p$-adic Markov chain corresponding to $\boldsymbol{l}$ given by (\ref{gt}).

We have proved the following theorem.
\begin{thm} Let $\mathcal T$ be a tree  containing a Cayley tree $\Gamma^k$ of order $k\geq 1$, as a subtree. Suppose the conditions of Theorem \ref{t3} are satisfied on $\Gamma^k$. If (\ref{qcc}) is satisfied,
then on the tree $\mathcal T$ there are at least two $p$-adic Markov chains (one is $\mu^{\boldsymbol{l}}$ and the other is  $\mu^{\boldsymbol{1}}$).
\end{thm}

\section{{Criterion for the (un-)boundedness of the $p$-adic Markov chains}}

Now we are interested in finding out whether a $p$-adic Markov chain is bounded.

Let $\{\boldsymbol{z}(x,y)\in \mathcal E_p, \ \ \langle x, y \rangle\in L\}$  be a boundary law for the matrix $\mathcal Q_b=(Q_b(i,j))$
and $\mu^{\boldsymbol{z}}$
be the corresponding $p$-adic Markov chain.

\begin{thm} The following hold
\begin{itemize}
\item[1)] if  $\max_{i,j\in \Phi}|Q_b(i,j)|_p\leq |q|_p$ for all $b\in L$,
then the $p$-adic Markov chain $\mu^{\boldsymbol{z}}$ is bounded;
\item[2)]
if $\min_{i}\max_{j}|Q_b(i,j)|_p>|q|_p$ for all $b\in L$,
then the $p$-adic Markov chain $\mu^{\boldsymbol{z}}$ is not bounded.

\end{itemize}
\end{thm}


\begin{proof} It suffices to show that
for any finite \emph{connected} set
$\varnothing\ne\varLambda\subset V$ (denote
$\bar{\varLambda}=\varLambda\cup\partial\varLambda$), and any $\varsigma\in \Omega_\Lambda$, one has $|\mu^{\boldsymbol{z}}(\sigma_{\myn\bar{\varLambda}}\myn=\varsigma)|_p\leq M$, for some $M>0$.
Using (\ref{eq:mu-h-ex-Lambda}), we get
\begin{equation}\label{ea}
\left|\mu^{\boldsymbol{z}}(\sigma_{\myn\bar{\varLambda}}\myn=\varsigma)\right|_p=\left|{\prod_{x\in \partial \Lambda}
{z}_{\varsigma(x)}(x,x_\Lambda)\prod_{{b\in L:\atop b\cap \Lambda\ne \emptyset}}Q_b(\varsigma_b)\over
\sum_{\varphi_{\myn\bar{\varLambda}}}\prod_{x\in \partial \Lambda}
{z}_{\varphi(x)}(x,x_\Lambda)\prod_{{b\in L:\atop b\cap \Lambda\ne \emptyset}}Q_b(\varphi_b)}\right|_p.
\end{equation}
Let us calculate
\begin{align*}\mathcal Z&=\left|\sum_{\varphi_{\myn\bar{\varLambda}}}\prod_{x\in \partial \Lambda}
{z}_{\varphi(x)}(x,x_\Lambda)\prod_{{b\in L:\atop b\cap \Lambda\ne \emptyset}}Q_b(\varphi_b)\right|_p\\
&=\left|\sum_{\varphi_{\myn\bar{\varLambda}}}\left[\prod_{x\in \partial \Lambda}
{z}_{\varphi(x)}(x,x_\Lambda)-1\right]\prod_{{b\in L:\atop b\cap \Lambda\ne \emptyset}}Q_b(\varphi_b)+\sum_{\varphi_{\myn\bar{\varLambda}}}\prod_{{b\in L:\atop b\cap \Lambda\ne \emptyset}}Q_b(\varphi_b) \right|_p.
\end{align*}
The set $\bar{\varLambda}$ can be decomposed as
$$\bar{\varLambda}=\partial \Lambda\cup \partial_{int}\Lambda\cup \partial_{int}(\Lambda\setminus \partial_{int}\Lambda)\cup\dots\cup\{x_0\},$$
where $\partial_{int}A=\{x\in A: \exists y\in V\setminus A, \ \langle x, y \rangle\}$.
Since $\mathcal Q_b$ is stochastic for any $b\in L$ we get
\begin{align*}
&\sum_{\varphi_{\myn\bar{\varLambda}}}\prod_{{b\in L:\atop b\cap \Lambda\ne \emptyset}}Q_b(\varphi_b)=
\sum_{\varphi_{\myn{\varLambda}}}\prod_{{b\in L:\atop b\subset \Lambda\times \Lambda}}Q_b(\varphi_b)
\prod_{{b=\langle x, y\rangle\in L:\atop x\in \partial_{int}\Lambda, \, y\in \partial\Lambda}}\sum_{\varphi(y)}Q_b(\varphi(x),\varphi(y))\\
=&\sum_{\varphi_{\Lambda\setminus \partial_{int}\Lambda}}\prod_{{b\in L:\atop b\subset (\Lambda\setminus \partial_{int}\Lambda)\times (\Lambda\setminus \partial_{int}\Lambda)}}Q_b(\varphi_b)
\prod_{{b=\langle x, y\rangle\in L:\atop  x\in \partial_{int}(\Lambda\setminus \partial_{int}\Lambda),\, y\in \partial_{int}\Lambda}}\sum_{\varphi(y)}Q_b(\varphi(x),\varphi(y))\\
=&\dots=\sum_{\varphi(x_0)=1}^q1=q.
\end{align*}

 1) Under the conditions of the part 1), we have (note that $|q|_p\leq 1$)
$$\mathcal Z=\left|\sum_{\varphi_{\myn\bar{\varLambda}}}\left[\prod_{x\in \partial \Lambda}
{z}_{\varphi(x)}(x,x_\Lambda)-1\right]\prod_{{b\in L:\atop b\cap \Lambda\ne \emptyset}}Q_b(\varphi_b)+q\right|_p=|q|_p.
$$
Thus
 \begin{equation}\label{ea1}
\left|\mu^{\boldsymbol{z}}(\sigma_{\myn\bar{\Lambda}}\myn=\varsigma)\right|_p=\mathcal Z^{-1}\left|\prod_{x\in \partial \Lambda}
{z}_{\varsigma(x)}(x,x_\Lambda)\prod_{{b\in L:\atop b\cap \Lambda\ne \emptyset}}Q_b(\varsigma_b)\right|_p\leq {|q|_p^{|\bar{\Lambda}|}\over |q|_p}\leq 1.
\end{equation}

2) Suppose now the conditions of part 2) are satisfied. For a marginal on the two-site volume, i.e.,
an edge $b=\langle x, y\rangle$, corresponding to a boundary law $\boldsymbol{z}=(z_1(x,y),\dots,z_{q}(x,y))$,
when $\sigma(x)=i$ is fixed we have
$$\mu_b^{\boldsymbol{z}}(i,\sigma(y))={Q_b(i, \sigma(y))z_{\sigma(y)}(x,y)\over \sum_{\varphi(y)=1}^qQ_b(i, \varphi(y))z_{\varphi(y)}(x,y)}.$$
Therefore,
\begin{align}\label{de}
\begin{split}
\left|\mu_b^{\boldsymbol{z}}(i,\sigma(y))\right|_p=& \left|{Q_b(i, \sigma(y))z_{\sigma(y)}(x,y)\over \sum_{\varphi(y)=1}^q[z_{\varphi(y)}(x,y)-1]Q_b(i, \varphi(y))+\sum_{\varphi(y)=1}^qQ_b(i,\varphi(y))}\right|_p\\
=& {|Q_b(i, \sigma(y))|_p\over \left|\sum_{\varphi(y)=1}^q[z_{\varphi(y)}(x,y)-1]Q_b(i, \varphi(y))+q\right|_p}.
\end{split}
\end{align}

In order to show that the measure $\mu^{\boldsymbol{z}}$ is not bounded, it is enough to show that
its marginal measure is not bounded. Let $\pi=\{x_0, x_1, . . . \}$ be an arbitrary
infinite path in the tree. The marginal measure $\mu^{\boldsymbol{z}}_\pi$ has the form
\begin{equation}\label{m}
\mu^{\boldsymbol{z}}_\pi(\omega_n)=\alpha_{\omega_n(x_{0})}\prod^{n-1}_{m=0}
\mu_{\langle x_m, x_{m+1}\rangle}^{\boldsymbol{z}}(\omega_n(x_m),\omega_n(x_{m+1})).
\end{equation}
Here $\omega_n : \{x_0, ..., x_n\}\to \Phi=\{1,2,\dots,q\}$ is a configuration on
$\{x_0, . . . , x_n\}$ and $\alpha_i$ is a coordinate of the invariant stochastic vector of the
matrix $\left(\mu_{\langle x_0, x_1\rangle}^{\boldsymbol{z}}(i, j)\right)_{i,j=1,\dots,q}$.

To ensure that $|\alpha_{\omega_n(x_0)}|_p>c$ for some $c>0$.
 We can choose the value $i_0=\omega_n(x_{0})$ (of the configuration $\omega_n$ on the vertex $x_{0}$) such that
$$|\alpha_{i_0}|_p=\max_{s\in\Phi} |\alpha_{s}|_p.$$
Then since $\alpha$ is a probability vector we have
$$1=\left|\sum_{s=1}^q \alpha_{s}\right|_p \leq   \max_{s\in\Phi} |\alpha_{s}|_p=|\alpha_{\omega_n(x_{0})}|_p.$$
Having $i_0$, we choose the value $i_1=\omega_n(x_1)$ of the configuration $\omega_n$ to satisfy
$$|Q_b(i_0,i_1)|_p=\max_{j}|Q_b(i_0,j)|_p.$$
By iterating, we define $i_m=\omega(x_m)$  to have
$$|Q_b(i_{m-1}, i_m)|_p=\max_{j}|Q_b(i_{m-1},j)|_p, \ \ m\geq 1.$$
 Then for the above constructed $\omega_n$, by (\ref{de}) we get
\begin{align}\label{ll}
\begin{split}
\left|\mu_{\langle x_m, x_{m+1}\rangle}^{\boldsymbol{z}}(i_m, i_{m+1})\right|_p=&
{\left|Q_b(i_m, i_{m+1})\right|_p\over  \left|\sum_{j=1}^q[z_{j}(x,y)-1]Q_b(i_m, j)+q\right|_p}\\
\geq& {\left|Q_b(i_m, i_{m+1})\right|_p\over \max \left\{{1\over p}\max_{j}|Q_b(i_m, j)|_p, |q|_p\right\}}\\
=&{|Q_b(i_m, i_{m+1})|_p
\over \max\left\{{1\over p}|Q_b(i_m, i_{m+1})|_p, \ \ |q|_p\right\}}\geq p, \ \ m=1,2,\dots.
\end{split}
\end{align}
Here, at the last step we have used the following (which is true by the condition of the part 2) of theorem)
$$|q|_p<|Q_b(i_m, i_{m+1})|_p.$$
Consequently, for such a configuration $\omega_n$, from (\ref{m}) and  (\ref{ll}), we find that
$$\mu^{\boldsymbol{z}}_\pi(\omega_n)\geq p^n,$$
i.e., $\mu^{\boldsymbol{z}}$ is not bounded.
\end{proof}

\section*{ Acknowledgements}

UAR thanks the University Paris-Est Cr\'eteil (UPEC) for the hospitality during June 2019, where this work has been achieved, and Labex B\'ezout (Universit\'e Paris Est) for the financial and logistic support of this visit. The collaboration of the authors is realized within the project "Real/ $p$-adic dynamical systems and Gibbs measures" funded by LabEx B\'ezout (ANR-10-LABX-58).

\end{document}